\documentclass[11pt, letterpaper]{article}
\usepackage
{
      amssymb,
      amsthm,
      amsmath,
      hhline,
      fullpage,
      float,
      xspace,
      nicefrac,
}

\usepackage{prettyref}

\newrefformat{eq}{\textup{(\ref{#1})}}
\newrefformat{lem}{Lemma~\ref{#1}}
\newrefformat{cor}{Corollary~\ref{#1}}
\newrefformat{claim}{Claim~\ref{#1}}
\newrefformat{thm}{Theorem~\ref{#1}}
\newrefformat{cha}{Chapter~\ref{#1}}
\newrefformat{sec}{Section~\ref{#1}}
\newrefformat{tab}{Table~\ref{#1} on page~\pageref{#1}}
\newrefformat{fig}{Figure~\ref{#1} on page~\pageref{#1}}
\newrefformat{def}{Definition~\ref{#1}}

\newcommand {\ONE}     {I}

\newcommand {\tG} {{\widetilde{G}}}
\newcommand {\tH} {{\widetilde{H}}}

\DeclareMathOperator {\vol}  {vol}
\DeclareMathOperator {\argmax} {argmax}

\DeclareMathOperator {\VALUE}  {value}
\DeclareMathOperator {\OPT}  {OPT}

\newcommand {\roundup}   [1] {{\lceil {#1} \rceil}}
\newcommand {\rounddown} [1] {{\lfloor {#1} \rfloor}}

\newcommand {\set}   [1] {\left\{ #1 \right\}}

\DeclareMathOperator {\Exp}       {\mathbb{E}}

\newcommand {\Prob}  [1] {\Pr \left\{#1 \right\}}

\newcommand {\E}     [1] {\mathbb{\Exp}\left[#1\right]}

\newcommand{\given}{\;\mid\;}
\newcommand{\Given}{\;\;\mid\;\;}

\floatstyle{ruled}
\newfloat{algor}{thp}{lop}
\floatname{algor}{Algorithm}

\newcommand {\bbR}    {\mathbb{R}}

\newcommand {\calC}   {{\cal{C}}}

\newcommand {\calE}   {{\cal{E}}}

\newcommand {\calW}   {{\cal{W}}}

\newcommand {\calR}   {{\cal{R}}}
\newcommand {\calS}   {{\cal{S}}}
\newcommand {\calP}   {{\cal{P}}}

\newcommand {\NP}     {{\cal{NP}}}
\newcommand {\QP}     {{\cal{QP}}}
\newcommand {\BPQP}     {{\cal{BPQP}}}

\newcommand{\R}{\mathbb{R}}

\newcommand {\OPTQAP} {\OPT_{\text{\tiny{QAP}}}}
\newcommand {\VALUEQAP} {\VALUE_{\text{\tiny{QAP}}}}

\newcommand {\OPTLC} {\OPT_{\text{\tiny{LC}}}}
\newcommand {\VALUELC} {\VALUE_{\text{\tiny{LC}}}}

\newcommand {\MAXQAP} {\textsc{MaxQAP}\xspace}
\newcommand{\defeq}{\stackrel{\mathrm{def}}=}

\newtheorem{theorem}{Theorem}[section]

\newtheorem{corollary}[theorem]{Corollary}

\newtheorem{lemma}[theorem]{Lemma}
\newtheorem{remark}[theorem]{Remark}

\newtheorem{definition}[theorem]{Definition}
\newtheorem{claim}[theorem]{Claim}

\title {Maximum Quadratic Assignment Problem:\\
\large{Reduction from Maximum Label Cover and}\\
\large{LP-based Approximation Algorithm\footnote{The conference version of the paper appeared at ICALP 2010.}}}

\author {Konstantin Makarychev\footnote{Microsoft Research, Redmond,
    WA 98122, USA}%
  \and Rajsekar Manokaran\footnote{School of CSC, KTH, Stockholm, Sweden.  Work done while visiting
    IBM TJ Watson Research Center, NY, USA.}%
  \and Maxim Sviridenko\footnote{Yahoo! Labs}} \date{}
\begin{document}
\maketitle

\begin{abstract}
  We show that for every positive $\varepsilon > 0$, unless $\NP
  \subset \BPQP$, it is impossible to approximate the maximum
  quadratic assignment problem within a factor better than
  $2^{\log^{1-\varepsilon} n}$ by a reduction from the maximum label
  cover problem. Our result also implies that Approximate Graph
  Isomorphism is not robust and is in fact, $1 - \epsilon$ vs
  $\epsilon$ hard assuming the Unique Games Conjecture.
  
  Then, we present an $O(\sqrt{n})$-approximation
  algorithm for the problem based on rounding of the linear
  programming relaxation often used in the state of the art exact
  algorithms.
\end{abstract}

\section{Introduction}
In this paper we consider the Quadratic Assignment Problem.  An
instance of the problem, $\Gamma = (G, H)$ is specified by two
weighted graphs $G = (V_G, w_G)$ and $H = (V_H, w_H)$ such that $|V_G|
= |V_H|$ (we denote $n=|V_G|$). The set of feasible solutions consists
of bijections from $V_G$ to $V_H$. For a given bijection $\varphi$ the
objective function is

\begin{equation} \VALUEQAP(\Gamma, \varphi) = \sum_{(u,v) \in V_G
    \times V_G} w_G(u,v) w_H(\varphi(u), \varphi(v)). \label{objective}
\end{equation}

There are two variants of the problem the Minimum Quadratic Assignment
Problem and the Maximum Quadratic Assignment Problem (\MAXQAP)
depending on whether the objective function (\ref{objective}) is to be
minimized or maximized. The problem was first defined by Koopmans and
Beckman~\cite{KB} and sometimes this formulation of the problem is
referred to as Koopmans-Beckman formulation of the Quadratic
Assignment Problem.  Both variants of the problem model an
astonishingly large number of combinatorial optimization problems such
as traveling salesman, maximum acyclic subgraph, densest subgraph and
clustering problems to name a few. It also generalizes many practical
problems that arise in various areas such as modeling of backboard
wiring~\cite{S}, campus and hospital layout~\cite{DH,E},
scheduling~\cite{GG} and many others~\cite{EL,LM}. The surveys and
books~\cite{A,bcpp,c,BDM,labhq,pw} contain an in-depth treatment of
special cases and various applications of the Quadratic Assignment
Problem.

The Quadratic Assignment Problem is an extremely difficult
optimization problem. The state of the art exact algorithms can solve
instances with approximately 30 vertices, so a lot of research effort
was concentrated on constructing good heuristics and relaxations of
the problem.

{\bf Previous Results. } The Minimum Quadratic Assignment Problem is
known to be hard to approximate even under some very restrictive
conditions on the weights of graphs $G$ and $H$.  In particular, even
when $H$ induces a line metric, any polynomial factor approximation
(in polynomial time) implies that $\calP = \NP$~\cite{Q}.  Polynomial
time exact~\cite{c} and approximation algorithms~\cite{hls} are known
for very specialized instances.

In contrast, \MAXQAP seem to be more tractable. Barvinok \cite{B}
constructed an approximation algorithm with performance guarantee
$\varepsilon n$ for any $\varepsilon>0$. Nagarajan and
Sviridenko~\cite{NS} designed $O(\sqrt{n}\log^2n)$-approximation
algorithm by utilizing approximation algorithms for the minimum
vertex cover, densest $k$-subgraph and star packing problems. For the
special case when one of the edge weight functions ($w_G$ or $w_H$)
satisfy the triangle inequality there are combinatorial
$4$-approximation~\cite{AHS} and LP-based $3.16$-approximation
algorithms~\cite{NS}. Another tractable special case is the so-called
dense Quadratic Assignment Problem~\cite{AFK}. This special case
admits a sub-exponential time approximation scheme and in some cases
it could be implemented in polynomial time \cite{AFK,FK}. On the
negative side, APX-hardness of \MAXQAP is implied by the APX-hardness
of its special cases, e.g. Traveling Salesman Problem with Distances
One and Two~\cite{PY}.

An interesting special case of \MAXQAP is the Densest $k$-Subgraph
Problem. The best known algorithm by Bhaskara, Charikar, Chlamtac,
Feige, and Vijayaraghavan~\cite{BCCFV} gives a $O(n^{1/4})$
approximation.  However, the problem is not even known to be APX-hard
(under standard complexity assumptions).  Feige~\cite{F02} showed that
the Densest $k$-Subgraph Problem does not admit a $\rho$-approximation
(for some universal constant $\rho > 1$) assuming that random 3-SAT
formulas are hard to refute.  Khot~\cite{Khot} ruled out PTAS for the
problem under the assumption that $\NP$ does not have randomized
algorithms that run in sub-exponential time.

{\bf Our Results. } Our first result is the first
superconstant non-approximability for \MAXQAP.
We show that for every positive $\varepsilon > 0$, unless
$\NP \subset \BPQP$ ($\BPQP$ is
the class of problems solvable in randomized
quasipolynomial-time), it is impossible to
approximate the maximum quadratic assignment problem with the
approximation factor better than $2^{\log^{1-\varepsilon} n}$.
Particularly, there is no  polynomial time
poly-logarithmic approximation algorithms for \MAXQAP under the above complexity assumption.
It is an interesting open
question if our techniques can be used to obtain a similar result for
the Densest $k$-Subgraph Problem.

Our second result is an $O(\sqrt{n})$-approximation algorithm based on
rounding of the optimal solution of the linear programming relaxation.
The LP relaxation was first considered by Adams and Johnson~\cite{AJ}
in 1994.
As a consequence of our result we obtain a bound of $O(\sqrt{n})$ on
the integrality gap of this relaxation that almost matches a lower
bound of $\Omega(\sqrt{n}/\log n)$ of~Nagarajan and
Sviridenko~\cite{NS}.  Note, that the previous
$O(\sqrt{n}\log^2n)$-approximation algorithm~\cite{NS} was not based
on the linear programming relaxation, and therefore no non-trivial
upper bound on the integrality gap of the LP was known.

{\bf Note Added in Proof.} Suppose that the graphs $G$ and $H$ have the same number of
edges. Then, $G$ and $H$ are isomorphic if and only if the optimal value of the unweighted
Maximum Quadratic Assignment problem equals 1. This observation gives another name to the problem: The unweighted version of Maximum Quadratic
Assignment is also known as~\textit{Approximate Graph Isomorphism}. In 
Approximate Graph Isomorphism, it is natural to divide
the objective function by $|E_G|=|E_H|$, then for isomorphic graphs $G$ and $H$, the optimal
objective value is 1. We do not know the complexity of the (exact)
Graph Isomorphism problem, and hence we do not know whether finding the exact solution 
for satisfiable instances of Approximate Graph Isomorphism (i.e., instances of value $1$) is 
$\NP$-hard. In several recent works~\cite{AKKV, DWWZ} (published after the conference version of this paper appeared at ICALP~2010), the authors asked what can be done if the instance is almost satisfiable i.e., the value of the 
optimal solution is at least $(1-\varepsilon)$. The immediate corollary of our result is that it 
is not possible to distinguish instances of value at least $(1-\varepsilon)$ and instances of value at most $\delta$ 
in randomized polynomial time for every positive $\varepsilon$ and $\delta$. This result holds assuming that the (randomized) Unique Games Conjecture holds. In other words, we assume that for every positive $\varepsilon$ and $\delta$,
there is no randomized polynomial-time algorithm that distinguishes $(1-\varepsilon)$ satisfiable instances of 
Unique Games and $\delta$ satisfiable instances of Unique Games. To get the result, in the reduction we present below, we need  to use an instance of MAX $\Gamma$-Lin$(k)$ instead of an arbitrary instance of Label Cover, 
the graph $G$ contains $k$ copies of the constraint graph of the MAX $\Gamma$-Lin$(k)$ instance, the graph $H$ is 
the label-extended graph of the MAX $\Gamma$-Lin$(k)$ instance.

\section{Hardness of Approximation}

A weighted graph $G = (V, w)$ is specified by a vertex set $V$ along
with a weight function $w: V \times V \to \R$ such that for every $u,
v \in V$, $w(u, v) = w(v, u)$ and $w(u,v)\geq 0$.  An edge $e = (u, v)$ is said to be
present in the graph $G$ if $w(u, v)$ is non-zero.

We prove the inapproximability of the \MAXQAP problem via an
approximation preserving poly-time randomized reduction from the Label
Cover problem defined below.

\begin{definition}[Label Cover Problem]
  An instance of the label cover problem denoted by $\Upsilon = (G =
  (V_G, E_G), \pi, [k])$ consists of a graph $G$ on $V_G$ with edge
  set $E_G$ along with a set of labels $[k] = \{0, 1, \ldots k-1\}$.
  For each edge $(u, v) \in E_G$, there is a constraint $\pi_{uv}$, a
  subset of $[k] \times [k]$ defining the set of accepted labelings
  for the end points of the edge.  The goal is to find a labeling of
  the vertices, $\Lambda: V_G \to [k]$ maximizing the total fraction
  of the edge constraints satisfied. We will denote the optimum of an
  instance $\Upsilon$ by $\OPTLC(\Upsilon)$.  In other words,
  \[ \OPTLC (\Upsilon) \defeq \max_{\Lambda: V_G \to
    [k]}\frac{1}{|E_G|}\sum_{(u,v) \in E}\ I((\Lambda(u), \Lambda(v))
  \in \pi_{uv}), \]where $\ONE (\cdot)$ is the indicator of an event.
  We denote the optimum by $\OPTQAP(\Gamma)$.  We will denote the
  fraction of edges satisfied by a labeling $\Lambda$ by
  $\VALUELC(\Upsilon, \Lambda)$.
\end{definition}


The PCP theorem~\cite{PCP1,PCP2}, along with the Raz parallel repetition theorem~\cite{Raz} shows that
the label cover problem is hard to approximate within
a factor of $2^{\log^{1-\varepsilon}{n}}$.

\begin{theorem}[see e.g., Arora and Lund~\cite{PCP0}]\label{thm:lc-hardness}
If $\NP\not\subset \QP$, then for every positive $\varepsilon > 0$,
it is not possible to distinguish satisfiable instances of the label
cover problem from instances with optimum
at most $2^{-\log^{1-\varepsilon}{n}}$ in polynomial time.
\end{theorem}

We will show an approximation preserving reduction from a label cover
instance to a \MAXQAP instance such that: If the label cover instance
$\Upsilon$ is completely satisfiable, then the \MAXQAP instance $\Gamma$
will have optimum $1$; on the other hand, if $\OPTLC(\Upsilon)$ is at
most $\delta$, then no bijection $\varphi$ obtains a value greater
than $O(\delta)$.

Strictly speaking, the problem is not well defined when the graphs $G$
and $H$ do not have the same number of vertices.  However, in our
reduction, we will relax this condition by letting $G$ have fewer
vertices than $H$, and allowing the map $\varphi$ to be only injective
(i.e., $\varphi(u) \neq \varphi(v)$, for $u \neq v$).  The reason is
that we can always add enough isolated vertices to $G$ to satisfy
$|V_G| = |V_H|$. We also assume that the graphs are unweighted, and
thus given an instance $\Gamma$ consisting of
two graphs $G= (V_G,E_G)$ and $H=(V_H,E_H)$, the goal is to find an injective map $\varphi: V_G \to V_H$, so as maximize
$$\VALUEQAP (\Gamma, \varphi) = \sum_{(u,v)\in E_G} \ONE ((\varphi(u),\varphi(v))\in E_H). $$

Informally, our reduction does the following.  Given an instance
$\Upsilon = (G = (V_G, E_G), \pi, [k])$ of the label cover problem,
consider the \emph{label extended} graph $H$ on $V_G \times [k]$ with
edges $((u, i)-(v, j))$ for every $(u, v) \in E_G$ and every accepting
label pair $(i, j) \in \pi_{uv}$.  Every labeling $\Lambda$ for
$\Upsilon$ naturally defines an injective map, $\varphi$ between $V_G$
and $V_G \times [k]$: $\varphi(u) = (u, \Lambda(u))$.  Note that
$\varphi$ maps edges satisfied by $\Lambda$ onto edges of $H$.
Conversely, given an injection $\varphi: V_G \to V_G \times [k]$ such that
$\varphi(u) \in \{u\} \times [k]$ for every $u \in V_G$, we can
construct a labeling $\Lambda$ for $\Upsilon$ satisfying exactly the
constraint edges in $G$ which were mapped on to edges of $H$.
However, the requirement that $\varphi(u) = (u, \Lambda(u))$
is crucial for the converse to hold: an arbitrary injective map might not correspond
to any labeling of the label cover $\Upsilon$.

To overcome the above shortcoming, we modify the graphs $G$ and $H$ as
follows.  We replace each vertex $u$ in $G$ with a ``cloud'' of
vertices $\{(u,i): i \in [N]\}$ and each vertex $(u,x)$ in $H$ with a
cloud of vertices $\{(u,x,i): i \in [N]\}$, each index $i$ is from a
significantly large set $[N]$.  Call the new graphs $\tG$ and $\tH$
respectively.

For every edge $(u, v) \in E_G$, the corresponding clouds in $\tG$ are
connected by a random bipartite graph where each edge occurs with
probability $\alpha$.  We do this independently for each edge in $E_G$.
For every accepting pair $(x, y) \in \pi_{uv}$, we copy the
``pattern'' between the clouds $(u, x, \star)$ and $(v, y, \star)$ in
$\tH$.

As before, every solution of the label cover problem $u \mapsto
\Lambda(u)$ corresponds to the map $(u,i) \mapsto
(u,\Lambda(u),i)$ which maps every ``satisfied'' edge of $\tilde{G}$
to an edge of $\tilde{H}$. However, now, we may assume that every $(u, i)$ is
mapped to some $(u, x, i)$, since, loosely speaking, the pattern of
edges between $(u, \star)$ and $(v, \star)$ is unique for each edge
$(u, v)$: there is no way to map the \textit{cloud} of $u$ to the
\textit{cloud} of $u'$ and the \textit{cloud} of $v$ to the
\textit{cloud} of $v'$ (unless $u = u'$ and $v = v'$), so that more
than an $\alpha$ fraction of the edges of one cloud are mapped on
edges of the other cloud.  We will make the above discussion formal in
the rest of this section.


\rule{0pt}{12pt}
\hrule height 0.8pt
\rule{0pt}{1pt}
\hrule height 0.4pt
\rule{0pt}{6pt}

\noindent\textbf{Hardness Reduction}

    \smallskip
    \noindent\textbf{Input}: A label cover instance $\Upsilon
    = (G = (V_G, E_G), \pi, [k])$.%

    \smallskip%
    \noindent\textbf{Output}: A \MAXQAP instance $\Gamma = (\tG, \tH)$;
    $\tG = (V_\tG, E_\tG)$, $\tH = (V_\tH, E_\tH)$.%

    \smallskip %
    \noindent\textbf{Parameters}: Let $N = \roundup{n^4|E_G|k^5}$ and $\alpha = \nicefrac{1}{n}$.
\begin{itemize}
 \item Define $V_{\tG} = V_G \times [N]$ and $V_{\tH} = V_G \times [k] \times [N]$.%
 \item For every edge $(u,v)$ of $G$ pick a random set of pairs
   $\calE_{uv}\subset [N]\times[N]$. Each pair $(i,j)\in [N]\times
   [N]$ belongs to $\calE_{uv}$ independently with probability $\alpha$.
 \item For every edge $(u,v)$ of $G$ and every pair $(i,j)$ in
   $\calE_{uv}$, add an edge $((u,i), (v,j))$ to $\tG$. Then
$$E_{\tG} = \{((u,i), (v,j)): (u,v) \in E_G \text{ and } (i, j) \in \calE_{uv}\}.$$
\item For every edge $(u,v)$ of $G$, every pair $(i,j)$ in
  $\calE_{uv}$, and every pair $(x,y)$ in $\pi_{uv}$, add an edge
  $((u,x, i), (v, y,j))$ to $\tH$. Then
$$E_{\tH} = \{((u,x,i), (v,y,j)): (u,v) \in E_G \text{, } (i, j) \in \calE_{uv} \text{ and }
(x,y)\in\pi_{uv}\}.$$
\end{itemize}
\rule{0pt}{1pt}
\hrule height 0.4pt
\rule{0pt}{1pt}
\hrule height 0.8pt
\rule{0pt}{12pt}

\bigskip
It is easy to see that the reduction runs in polynomial time.
In our reduction, both $k$ and $N$ are polynomial in $n$.

We will now show that the reduction is in fact approximation
preserving with high probability.  In the rest of the section, we will
assume $\Gamma = (\tG, \tH)$ is a \MAXQAP instance obtained from a
label cover instance $\Upsilon$ using the above reduction with
parameteres $N$ and $\alpha$.  Note that $\Gamma$ is a random variable.

We will first show that if the label cover instance has a good
labeling, then the \MAXQAP instance output by the above reduction has a
large optimum.  The following claim, which follows from a simple
concentration inequality, shows that the graph $\tG$ has, in fact, as
many edges as expected.

\begin{claim}\label{claim:conc-total} With high probability, $\tG$
  contains at least $\alpha |E_G| N^2 / 2$ edges. 
\end{claim}

\begin{lemma}[Completeness]\label{lem:completeness} Let $\Upsilon$ be
  a satisfiable instance of the Label Cover Problem. Then there exists
  a map of $\tG$ to $\tH$ that maps every edge of $\tG$ to an edge of
  $\tH$. Thus, $\OPTQAP(\Gamma) = |E_{\tG}|$.
\end{lemma}
\begin{proof}
  Let $u \mapsto \Lambda(u)$ be the solution of the label cover that
  satisfies all constrains. Define the map $\varphi : V_{\tG} \to
  V_{\tH}$ as follows $\varphi (u,i) = (u, \Lambda(u), i)$. Suppose
  that $((u,i), (v,j))$ is an edge in $\tG$. Then $(u,v)\in E_G$ and
  $(i,j)\in \calE_{uv}$. Since the constraint between $u$ and $v$ is
  satisfied in the instance of the label cover,
  $(\Lambda(u),\Lambda(v)) \in \pi_{uv}$. Thus, $((u, \Lambda(u), i),
  (v, \Lambda(v), j)) \in E_{\tH}$.
\end{proof}

Next, we will bound the optimum of $\Gamma$ in terms of the value of
the label cover instance $\Upsilon$.  We do this in two steps.  We
will first show that for a fixed map $\varphi$ from $V_\tG$ to $V_\tH$
the expected value of $\Gamma$ can be bounded as a function of the
optimum of $\Upsilon$.  Note that this is well defined as $V_\tG$ and
$V_\tH$ are determined by $\Upsilon$ and $N$ (and
independent of the randomness used by the reduction).  Next, we show
that the value is, in fact, tightly concentrated around the expected
value.  Then, we do a simple union bound over all possible $\varphi$
to obtain the desired result.  In what follows, $\varphi$ is a fixed
injective map from $V_\tG$ to $V_\tH$. Denote the first, second and third components of $\varphi$ by
  $\varphi_{V}$, $\varphi_{label}$ and $\varphi_{[N]}$ respectively.
  Then, $\varphi (u,i) = (\varphi_{V} (u,i), \varphi_{label} (u,i),
  \varphi_{[N]}(u,i))$.

\begin{lemma}\label{lem:bound-expectation} For every injective map $\varphi: V_\tG
  \to V_\tH$,
$$ \E{\VALUEQAP (\Gamma, \varphi)} \le \alpha |E_G| N^2 \times (\OPTLC(\Upsilon) + \alpha). $$
\end{lemma}
\begin{proof}
Define a probabilistic labeling of $G$ as follows:
\textit{for every vertex $u$, pick a random $i\in [N]$, and assign label
$\varphi_{label}(u,i)$ to $u$ i.e., set $\Lambda (u) = \varphi_{label}(u,i)$}. The expected value of the solution to the Label Cover problem equals
\begin{eqnarray*}
\Exp_{\Lambda}[\VALUELC (\Upsilon, \Lambda)] &=&  \frac{1}{|E_G|}\sum_{(u,v)\in E_G} \Exp_{\Lambda}
[\ONE ((\Lambda(u), \Lambda(v)) \in \pi_{uv})]\\
&=&
\frac{1}{|E_G|}\sum_{(u,v)\in E_G} \frac{1}{N^2}\sum_{i, j \in [N]} \ONE ((\varphi_{label}(u,i), \varphi_{label}(v,j)) \in \pi_{uv}).
\end{eqnarray*}
Since $\VALUELC (\Upsilon, \Lambda) \leq \OPTLC(\Upsilon)$ for every labeling $u\mapsto \Lambda(u)$,
\begin{equation}\label{eq:expLC}
\sum_{(u,v)\in E_G}\sum_{i, j \in [N]} \ONE ((\varphi_{label}(u,i), \varphi_{label}(v,j)) \in \pi_{uv}) \leq
|E_G|\cdot N^2\cdot \OPTLC(\Upsilon).
\end{equation}
On the other hand,
\begin{eqnarray}
\E{\VALUEQAP (\Gamma, \varphi)} &=&
\E{\sum_{((u,i),(v,j))\in E_{\tG}} \ONE ((\varphi(u,i), \varphi(v,j)) \in E_\tH)}\nonumber\\
&=&
\sum_{(u,v)\in E_G} \sum_{i, j \in [N]}
\Prob{(i,j) \in \calE_{uv} \text{ and } (\varphi(u,i), \varphi(v,j)) \in E_\tH}.\label{eq:finalSum}
\end{eqnarray}

  Recall, that the goal of the whole
  construction was to \emph{force} the solution to map
  each $(u,i)$ to $(u,\varphi_{label}(u,i),i)$. Let
  $\calC_\varphi$ denote the set of quadruples that satisfy this
  property:
  \[ \calC_{\varphi} = \{(u,i,v,j): (u,v)\in E_G \text{ and }
  \varphi(u,i) = (u,\varphi_{label}(u,i),i),\; \varphi(v,j) =
  (v,\varphi_{label}(v,j),j)\}\ .\]

\noindent  If $(u,i,v,j) \in \calC_{\varphi}$, then
\begin{align*}
  \Pr \{(i,j) \in \calE_{uv} \text{ and } (\varphi(u,i), \varphi(v,j))
  &\in E_\tH\}\\
  &=\Prob{(i,j) \in \calE_{uv} \text{ and }(\varphi_{label}(u,i),
    \varphi_{label}(v,j))\in \pi_{uv}}\\
  &=\Prob{(i,j) \in \calE_{uv}}\cdot \ONE((\varphi_{label}(u,i),
  \varphi_{label}(v,j))\in \pi_{uv})\\
  &= \alpha \cdot \ONE((\varphi_{label}(u,i), \varphi_{label}(v,j))\in
  \pi_{uv}).
\end{align*}

If $(u,v)\in E_G$, but $(u,i,v,j) \notin \calC_{\varphi}$, then either $(i,j) \neq  (\varphi_{[N]}(u,i), \varphi_{[N]}(v,j))$ or $(u,v)\neq (\varphi_V(u,i),\varphi_V(v,j))$, and hence the events $\{(i,j) \in
\calE_{uv}\}$ and $\{(\varphi_{[N]}(u,i), \varphi_{[N]}(v,j)) \in
\calE_{\varphi_V(u,i)\varphi_V(v,j)}\}$ are independent. We have
\begin{multline*}
  \Prob{(i,j) \in \calE_{uv} \text{ and } (\varphi(u,i), \varphi(v,j))
    \in E_\tH}\leq \\
  \Prob{(i,j) \in \calE_{uv} \text{ and } (\varphi_{[N]}(u,i),
    \varphi_{[N]}(v,j)) \in \calE_{\varphi_V(u,i)\varphi_V(v,j)}}\leq \alpha^2.
\end{multline*}
Now, splitting summation~(\ref{eq:finalSum}) into two parts depending on
whether $(u,i,v,j) \in \calC_\varphi$, we have
$$\E{\VALUEQAP (\Gamma, (\varphi)} \le \alpha |E_G| N^2 \OPTLC(\Upsilon) + \alpha^2 |E_G| N^2.$$
\end{proof}

We use the following concentration inequality for Lipschitz functions on the boolean cube.
\begin{theorem}[McDiarmid~\cite{MDiar}, Theorem~3.1, p.~206] \label{thm:mcdiarmid}
  Let $X_1,\dots X_T$ be independent random variables taking values in
  the set $\{0,1\}$.  Let $f: \{0,1\}^T \to \bbR$ be a $K$-Lipschitz
  function i.e., for every $x,y \in \{0,1\}^T$, $|f(x) - f(y)| \leq K
  \|x-y\|_1$. Finally, let $\mu = \E{f(X_1,\dots,X_T)}$.  Then for
   every positive $\varepsilon$,
$$\Prob{f(X_1,\dots,X_n) - \mu \geq \varepsilon} \leq e^{\frac{-2\varepsilon^2}{TK^2}}.$$
\end{theorem}

\begin{lemma}\label{lem:conc-value}
  For every injective map $\varphi: V_\tG \to V_\tH$,
  \[ \Prob{\VALUEQAP (\Gamma, \varphi) - \E{\VALUEQAP (\Gamma, \varphi)}
    \geq \alpha N^2} \leq e^{-n^2Nk}. \]
\end{lemma}
\begin{proof}

  The presence of edges in the random graphs $\tG$ and $\tH$ is
  determined by the random sets $\calE_{uv}$ (where $(u,v)\in E_G$).
  Thus, we can think of the random variable $\VALUEQAP (\Gamma,
  \varphi)$ as of function of the indicator variables $X_{uivj}$,
  where $X_{uivj}$ equals 1, if $(i,j)\in \calE_{uv}$; and 0,
  otherwise. To be precise, $\VALUEQAP (\Gamma, \varphi)$ equals
$$\sum_{\substack{(u,v) \in E_G\\i,j\in [N]}}
X_{uivj} X_{\varphi_V(u,i)\varphi_{[N]}(u,i)\varphi_V(v,j)\varphi_{[N]}(v,j)} \ONE
((\varphi_{label}(u,i),\varphi_{label}(v,j))\in
\pi_{\varphi_V(u,i)\varphi_V(v,j)}).
$$
Observe, that variables $X_{uivj}$ are mutually independent (we
identify $X_{uivj}$ with $X_{vjui}$). Each $X_{uivj} = 1$ with
probability $\alpha$. Finally, $\VALUEQAP (\Gamma, \varphi)$ is $(k^2
+1)$-Lipschitz as a function of the variables $X_{uivj}$. That is, if
we change one of the variables $X_{uivj}$ from 0 to 1, or from 1 to 0,
then the value of the function may change by at most $k^2+1$. This
follows from the expression above, since for every fixed $\varphi$,
each $X_{uivj}$ may appear in at most $k^2+1$ terms (reason: there is
one term $X_{uivj}
X_{\varphi_V(u,i)\varphi_{[N]}(u,i)\varphi_V(v,j)\varphi_{[N]}(v,j)}$
and at most $k^2$ terms $X_{u'i'v'j'}
X_{\varphi_V(u',i')\varphi_{[N]}(u',i')\varphi_V(v',j')\varphi_{[N]}(v',j')}$,
such that $\varphi (u',i') = (u,x,i)$ and $\varphi (v',j') = (v,y,j)$
for some $x,y\in [k]$, since $\varphi$ is an injective map).
McDiarmid's inequality with $T = N^2\cdot |E_G|$, $K = (k^2+1)$, and
$\varepsilon = \alpha N^2$, implies the statement of the lemma.
\end{proof}

\begin{corollary}[Soundness]\label{cor:soundness} With high
  probability, the reduction outputs an instance $\Gamma$ such that
  \[ \OPTQAP(\Gamma) \le \alpha |E_G| N^2 \times (\OPTLC(\Upsilon) + 2\alpha)\]
\begin{remark}
It is instructive to think, that $2\alpha \ll \OPTLC(\Upsilon)$.
\end{remark}
\end{corollary}
\begin{proof}
  The total number of maps from $V_G$ to $V_H$ is $(nNk)^{nN}$.  Thus,
  by the union bound, with probability $1 - o(1)$, for
  every injective mapping $\varphi: V_G \to V_H$:
$$\VALUEQAP (\Gamma, \varphi) - \E{\VALUEQAP (\Gamma, \varphi)}
\leq \alpha N^2.$$ Plugging in the bound for the expected value from
\prettyref{lem:bound-expectation} gives
\[ \OPTQAP(\Gamma) \le \alpha |E_G| N^2 \OPTLC(\Upsilon) + \alpha^2 |E_G|
N^2 + \alpha N^2. \]
\end{proof}

 \begin{theorem}
   For every positive $\varepsilon > 0$, there is no polynomial time
   approximation algorithm for the Maximum Quadratic Assignment problem
   with the approximation factor less than $D = 2^{\log^{1-\varepsilon}
   n}$ (where $n$ is the number of vertices in
   the graph) unless $\NP \subset \BPQP$.
 \end{theorem}
 \begin{proof}
   Assume to the contrary that there exists a polynomial time algorithm
   $A$ with the approximation factor less than $D = 2^{\log^{1-\varepsilon}
     n}$ for some positive $\varepsilon$. We use this
   algorithm to distinguish satisfiable instances of the label
   cover from at most $1/(4D)$--\,satisfiable instances in randomized polynomial
   time, which is not possible (if $\NP \not\subset \BPQP$) according to
   \prettyref{thm:lc-hardness}.

   Let $\Upsilon$ be an instance of the label cover. Using the reduction described above
   transform $\Upsilon$ to an instance of \MAXQAP $\Gamma$. Run the
   algorithm $A$ on $\Gamma$. \emph{Accept} $\Upsilon$, if the value $A(\Gamma)$ returned by the algorithm is at least $|E_{\tG}|/D$. \emph{Reject} $\Upsilon$, otherwise. By
   \prettyref{lem:completeness}, if $\Upsilon$ is satisfiable, then $\OPTQAP(\Gamma) = |E_{\tG}|$ and, hence
   $A(\Gamma) \geq |E_{\tG}|/D$. Thus we always accept satisfiable
   instances. On the other hand, if the instance $\Upsilon$ is
at most $1/(4D)$--\,satisfiable, then, by \prettyref{cor:soundness},
with high probability
$$\OPTQAP(\Gamma) \le \alpha |E_G| N^2 (\OPTLC(\Upsilon) + 2\alpha)
< |E_{\tG}|/D,$$
the second inequality follows from $|E_{\tG}| \geq \alpha |E_G| N^2 / 2$ (see \prettyref{claim:conc-total}). Therefore, with high probability, we reject $\Upsilon$.
\end{proof}

\section{LP Relaxation and Approximation Algorithm}
We now present a new $O(\sqrt{n})$ approximation algorithm slightly improving
on the result of Nagarajan and Sviridenko~\cite{NS}. The new algorithm
is surprisingly simple. It is based on a rounding of a natural LP relaxation. The LP relaxation
is due to  Adams and Johnson~\cite{AJ}. Thus we show that the integrality gap of the LP is
$O(\sqrt{n})$.

Consider the following integer program. We have assignment
variables $x_{up}$ between vertices of the two graphs that
are indicator variables
of the events ``\emph{$u$ maps to $p$}'', and variables $y_{upvq}$ that
are indicator variables of the events ``\emph{$u$ maps to $p$ and
$v$ maps to $q$}''. The LP relaxation is obtained by dropping
the integrality condition on variables.

\pagebreak

\rule{0pt}{12pt}
\hrule height 0.8pt
\rule{0pt}{1pt}
\hrule height 0.4pt
\rule{0pt}{6pt}

\noindent \textbf{LP Relaxation}

$$\begin{array}{lll}
\max & \displaystyle\sum_{\substack{u,v \in V_G\\p,q\in V_H}}  w_G(u,v) w_H(p,q) y_{upvq}&\\
&\sum_{p\in V_H} x_{up} = 1,& \text{for all }  u\in V_G;\\
&\sum_{u\in V_G} x_{up} = 1,& \text{for all }   p\in V_H;\\
&\sum_{u \in V_G} y_{upvq} = x_{vq}, & \text{for all }  v\in V_G, \;p,q\in V_H;\\
& \sum_{p\in V_H} y_{upvq} = x_{vq}, &\text{for all }  u,v\in V_G, \;q\in V_H;\\
&y_{upvq} = y_{vqup}, &\text{for all } u,v\in V_G,\; p,q\in V_H;\\
&x_{up}\in [0,1], &\text{for all }  u\in V_G, \; p\in V_H;\\
&y_{upvq}\in[0,1], &\text{for all } u\in V_G, \; p\in V_H.\\
\end{array}
$$
\rule{0pt}{1pt}
\hrule height 0.4pt
\rule{0pt}{1pt}
\hrule height 0.8pt
\rule{0pt}{12pt}


\rule{0pt}{12pt}
\hrule height 0.8pt
\rule{0pt}{1pt}
\hrule height 0.4pt
\rule{0pt}{6pt}

\noindent \textbf{Approximation Algorithm}
\begin{enumerate}
\item We  solve the LP relaxation and obtain an optimal solution
$(x^*,y^*)$. Then we pick random subsets of vertices $L_G\subset V_G$
and $L_H\subset V_H$ of size $\rounddown{n/2}$. Let $R_G = V_G\setminus L_G$
and $R_H = V_H\setminus L_H$. In the rest of the algorithm,
we will care only about edges going from $L_G$ to $R_G$ and from $L_H$ to $R_H$;
and we will ignore edges that completely lie in $L_G$, $R_G$, $L_H$ or $R_H$.
\item For every vertex $u$ in the set $L_G$, we pick a vertex $p$ in $L_H$ with probability $x^*_{up}$ and set $\widetilde{\varphi}(u) = p$ (recall that
$\sum_{p} x^*_{up} = 1$, for all $u$; with probability $1 - \sum_{p\in L_H} x^*_{up}$ we
do not choose any vertex for $u$). Then for every vertex $p\in L_H$, which
is chosen by at least one element $u$, we pick one of these
$u$'s uniformly at random; and set $\varphi (u) =  p$ (in other words,
we choose a random $u\in \widetilde{\varphi}^{-1} (p)$ and set $\varphi(u) = p$). Let
$\widetilde{L}_G\subset L_G$ be the set of all chosen $u$'s.
\item We now find a bijection $\psi: R_G \to R_H$ so as to maximize the
contribution we get from edges from $\widetilde{L}_G$ to $R_G$ i.e., to maximize the sum
\begin{equation}\label{eq:def-psi}
\sum_{\substack{u\in \widetilde{L}_G\\ v\in R_G}} w_G(u,v) w_H(\varphi(u), \psi(v)).
\end{equation}
This can be done, since the problem is equivalent to the maximum matching problem between the sets $R_G$ and $R_H$ where the weight of the edge from $v$ to $q$ equals
$$\sum_{u\in \widetilde{L}_G} w_G(u,v) w_H(\varphi(u), q).$$
\item Output the union of the maps $\varphi$, $\psi$ and an arbitrary
  bijection from $L_G \setminus \widetilde{L}_G$ to $L_H \setminus
  \varphi(\widetilde{L}_G)$.
\end{enumerate}
\rule{0pt}{1pt}
\hrule height 0.4pt
\rule{0pt}{1pt}
\hrule height 0.8pt
\rule{0pt}{12pt}

\subsection{Analysis of the Algorithm}
\begin{theorem}\label{thm:analys}
The approximation ratio of the algorithm is $O(\sqrt{n})$.
\end{theorem}
While the algorithm is really simple, the analysis is more involved.
Let $LP^*$ be the value of the LP solution. To prove that the algorithm
gives $O(\sqrt{n})$-approximation, it suffices to show that
\begin{equation}\label{eq:SuffToShow}
\E{\sum_{\substack{u\in L_G\\v\in R_G}} w_G(u,v) w_H (\varphi(u),\psi (v))}
\geq \frac{LP^*}{O(\sqrt{n})}.
\end{equation}

We split all edges of graph $G$ into two sets: heavy edges and light edges. For each vertex $u \in V_G$, let $\calW_u$ be the set of $\roundup{\sqrt{n}}$ vertices $v \in V_G$ with the largest weight $w_G(u,v)$.  Then,
$$LP^* = \sum_{\substack{u\in V_G\\v\in V_G \setminus \calW_u}}
\sum_{p,q\in V_H} y^*_{upvq} w_G(u,v) w_H(p,q) +
\sum_{\substack{u\in V_G\\v\in\calW_u}}
\sum_{p,q\in V_H} y^*_{upvq} w_G(u,v) w_H(p,q).$$
Denote the first term by $LP_I^*$ and the second by $LP_{II}^*$. Instead of
working with $\psi$, we explicitly define two new bijective maps $\nu_{I}$
and $\nu_{II}$ from $R_G$ to $R_H$ and prove, that
$$\E{\sum_{\substack{u\in \widetilde{L}_G\\v\in R_G}} w_G(u,v) w_H (\varphi(u),\nu_{I} (v))}
\geq {\frac{LP_I^*}{O(\sqrt{n})}};\text{ and }\;
\E{\sum_{\substack{u\in \widetilde{L}_G\\v\in R_G}} w_G(u,v) w_H (\varphi(u),\nu_{II} (v))}
\geq
{\frac{LP_{II}^*}{O(\sqrt{n})}}.$$
These two inequalities imply the bound we need, since the sum~\prettyref{eq:SuffToShow}
is greater than or equal to each of the sums above (by the choice of $\psi$; see~(\ref{eq:def-psi})).
Before we proceed, we state two simple lemmas we need later (see the appendix for the proofs).

\begin{lemma}\label{fact1}
Let $S$ be a random subset of a set $V$. Suppose that for $u\in V$,
all events $\{u'\in S\}$ where $u'\neq u$ are jointly independent of
the event $\{u\in S\}$. Let $s$ be an element of $S$ chosen uniformly at random (if $S=\varnothing$, then $s$ is not defined). Then $\Prob{u=s} \geq  \Prob{u\in S}/(\E{|S|} + 1)$.
\end{lemma}

\begin{lemma}\label{fact2}
Let $S$ be a random subset of a set $L$, and $T$ be a random subset of a set $R$.
Suppose that for $(l,r)\in L\times R$, all events $\{l'\in S\}$ where $l'\neq l$ and all events $\{r'\in T\}$ where $r'\neq r$ are jointly independent of the event $\{(l,r) \in S\times T\}$.
Let $s$ be an element of $S$ chosen uniformly at random, and let $t$ be an element of $T$ chosen uniformly at random.
Then, $$\Prob{(l,r) = (s,t)} \geq \frac{\Prob{(l,r)\in S\times T}}{(\E{|S|} + 1)\times (\E{|T|} + 1)}$$
(here $(s,t)$ is not defined if $S=\varnothing$ or $T=\varnothing$).
\end{lemma}

The first map $\nu_{I}$ is a random permutation between  $R_G$ and $R_H$.
Observe, that given subsets $L_G$ and $L_H$, the events $\{\widetilde{\varphi}(u)=p\}$ are mutually independent for different $u$'s and the expected size of $\widetilde{\varphi}^{-1}(p)$ is at most 1, here $\widetilde{\varphi}^{-1}(p)$ is
the preimage of $p$ (recall the map $\widetilde{\varphi}$ may have collisions, and hence
$\widetilde{\varphi}^{-1}(p)$ may contain more than one element). Thus, by Lemma~\ref{fact1} applied to the set $\widetilde{\varphi}^{-1}(p)\subset L_G$,
$$\Prob{\varphi(u) = p\Given L_G, L_H} \geq \Prob{\widetilde{\varphi}(u) = p\Given L_G, L_H}/2
=
\begin{cases}
x^*_{up}/2,&\text{if } u\in L_G \text{ and } p\in L_H;\\
0,& \text{otherwise}.
\end{cases}
$$
For every $u,v\in V_G$ and $p,q \in V_H$, let $\calE_{upvq}$ be the event $\{u\in L_G,v\in R_G, p\in L_H, q\in R_H\}$. Then,
$$\Prob{\calE_{upvq}} = \Prob{u \in L_G, v \in R_G, p \in L_H, q \in R_H} \geq \frac{1}{16}.$$
Thus, the probability that $\varphi (u) = p$ and $\nu_I (v) = q$ is $\Omega(x^*_{up}/n)$.
We have
\begin{eqnarray*}
\E{\sum_{\substack{u\in L_G\\v\in R_G}} w_G(u,v) w_H (\varphi(u),\nu_{I} (v))}
&\geq& \Omega(1) \times
\sum_{u,v\in V_G}\sum_{p,q\in V_H} \frac{x^*_{up}}{n}\;w_G(u,v)w_H(p,q) \\
&\geq& \Omega(1) \times \sum_{p,q\in V_H} w_H(p,q)\sum_{u\in V_G} x^*_{up} \sum_{v\in \calW_u}\frac{w_G(u,v)}{n}\\
&\geq& \Omega(1) \times \sum_{p,q\in V_H} w_H(p,q)\sum_{u\in V_G} x^*_{up} \frac{\min\{w_G(u,v): v \in \calW_u\}}{\sqrt{n}}.
\end{eqnarray*}
On the other hand, using $\sum_{v\in V_G} y^*_{upvq}/x^*_{up} = 1$, we get
\begin{eqnarray*}
LP_I^* &=& \sum_{p,q\in V_H} w_H(p,q) \sum_{u\in V_G} x^*_{up}
\left(\sum_{v\in V_G\setminus\calW_u} \frac{y^*_{upvq}}{x^*_{up}} w_G(u,v)\right) \\
&\leq&
\sum_{p,q\in V_H} w_H(p,q)\sum_{u\in V_G} x^*_{up} \max\{w_G(u,v): v \in V_G \setminus \calW_u\}
\\
&\leq&
\sum_{p,q\in V_H} w_H(p,q)\sum_{u\in V_G} x^*_{up} \min\{w_G(u,v): v \in \calW_u\}.
\end{eqnarray*}
We now define $\nu_{II}$. For every $v \in V_G$, let
$$l(v) = \argmax_{u\in V_G} \left\{\sum_{p,q\in V_H}w_G(u,v) w_H(p,q)y^*_{upvq}\right\}.$$
We say that $(l(v),v)$ is a heavy edge. For every $u\in L_G$, let
$$\calR_u = \set{v\in R_G: l(v) = u}.$$
All sets $\calR_{u}$ are disjoint subsets of $R_G$.
Note, that $l(v)$ does not depend on the partitioning $V_G=L_G\cup R_G$ and $V_H=L_H\cup R_H$,
but $\calR_{u}$ depends on $R_G$. We now define a map
$\widetilde{\nu}_{II}:\calR_u \to R_H$ independently for each $u$
for which $\widetilde{\varphi}(u)$ is defined (even if $\varphi(u)$ is not defined). For every $v\in \calR_u$, and $q\in R_H$, define
$$z_{vq} = \frac{y^*_{u\widetilde{\varphi}(u)vq}}{x^*_{u\widetilde{\varphi}(u)}}.$$
Observe, that $\sum_{v\in \calR_u}z_{vq} \leq 1$ for each $q\in R_H$ and $\sum_{q\in R_H}z_{vq} \leq 1$ for each $v\in \calR_u$. Hence, for a fixed $\calR_u$, the vector $(z_{vq}:v\in \calR_u, q\in R_H)$ lies in the convex hull of integral
partial matchings between $\calR_u$ and $R_H$. Thus,
the fractional matching $(z_{vq}:v\in \calR_u, q\in R_H)$ can be represented as a convex combination of integral partial matchings. Pick one of them with the probability proportional to its weight in the convex combination. Call this matching $\widetilde{\nu}^u_{II}$. Note, that $\widetilde{\nu}^u_{II}$
is injective and that the supports of $\widetilde{\nu}^{u'}_{II}$ and
$\widetilde{\nu}^{u''}_{II}$ do not intersect if $u'\neq u''$ (since $\calR_{u'}\cap\calR_{u''} =\varnothing$). Let $\widetilde{\nu}_{II}$ be
the union of $\widetilde{\nu}^u_{II}$ for all $u\in L_G$. The partial map
$\widetilde{\nu}_{II}$ may not be injective and may map several vertices of $R_G$ to the same vertex $q$. Thus, for every $q$ in the image
of $R_G$, we pick uniformly at random one preimage $v$ and set $\nu_{II} (v) = q$. We define $\nu_{II}$ on the rest of $R_G$ arbitrarily.

Fix $L_G$, $L_H$, $R_G = V_G\setminus L_G$, $R_H = V_H\setminus L_H$, and also $\calR_u = \set{v\in R_G: l(v) = u}$ (for all $u\in L_G$). Let $u\in L_G$, $v\in \calR_u$, $p\in L_H$ and $q \in R_H$. We want to estimate the probability that $\varphi(u) = p$ and $\nu_{II}(v) = q$.  Observe, that given sets $L_G$ and $L_H$, the event $\{\widetilde{\varphi}(u) = p\text{ and } \widetilde{\nu}_{II}(v) = q\}$ is independent of all events $\{\widetilde{\varphi}(u') = p\}$ for $u'\neq u$ and all events $\{\widetilde{\nu}_{II}(v') = q\}$ for $v'\notin \calR_u$. The expected size of $\widetilde{\nu}^{-1}_{II} (q)$ is at most 1, since
\begin{multline*}
\sum_{u'\in L_G}\sum_{v'\in \calR_{u'}} \Prob{\widetilde{\nu}^{u'}_{II} (v') = q} \leq
\sum_{u'\in L_G}\sum_{v'\in \calR_{u'}} \sum_{p'\in L_H} x^*_{u'p'} y^*_{u'p'v'q}/x^*_{u'p'} \leq \\
\sum_{v'\in V_G}\sum_{p'\in V_H} y^*_{l(v')p'v'q} =
\sum_{v'\in V_G} x^*_{v'q}\leq 1.
\end{multline*}
Therefore, by Lemma~\ref{fact2},
\begin{multline*}
\Prob{\varphi(u) = p\text{ and } \nu_{II}(v) = q \given L_G, L_H, u\in L_G, v\in \calR_u, p\in L_H, q \in R_H}
\geq \\ \Prob{\widetilde{\varphi}(u) = p\text{ and } \widetilde{\nu}_{II}(v) = q\given L_G, L_H, u\in L_G,
v\in \calR_u, p\in L_H, q \in R_H}/4 = y^*_{upvq}/4.
\end{multline*}

We are now ready to estimate the value of the solution:
\begin{eqnarray*}
\E{\sum_{\substack{u\in L_G\\v\in R_G}} w_G(u,v) w_H (\varphi(u),\nu_{II} (v))}
&\geq&
\Exp\limits_{L_G,L_H}\left[\sum_{\substack{u\in L_G\\v\in \calR_u}}\; \sum_{\substack{p\in L_H\\q\in R_H}} \frac{y^*_{upvq}}{4}
\;w_G(u,v)w_H(p,q)\right]\\
&=&\frac{1}{4}\Exp\limits_{L_G,L_H}\left[\sum_{\substack{v\in R_G: l(v)\in L_G}} \sum_{\substack{p\in L_H\\q\in R_H}} y^*_{l(v)pvq}
\;w_G(l(v),v)w_H(p,q)\right]\\
&=& \frac{1}{4} \sum_{v\in V_G}\sum_{p,q\in V_H}  \Prob{\calE_{l(v)pvq}}\;y^*_{l(v)pvq}\;w_G(l(v),v)w_H(p,q)\\
&=& \frac{1}{64} \sum_{v\in V_G}\sum_{p,q\in V_H} y^*_{l(v)pvq}\;w_G(l(v),v)w_H(p,q)\\
&\geq& \frac{1}{64} \sum_{v\in V_G}\max_{u\in V_G}
\left\{\sum_{p,q\in V_H} y^*_{upvq} \;w_G(u,v)w_H(p,q)\right\}\\
&\geq& \frac{1}{64}\sum_{v\in V_G}\frac{1}{|\calW_v|}\sum_{u \in \calW_v}
\left(\sum_{p,q\in V_H} y^*_{upvq} \;w_G(u,v)w_H(p,q)\right)\\
&=&\frac{1}{64} \times \frac{LP^*_{II}}{\roundup{\sqrt{n}\;}}.
\end{eqnarray*}

This finishes the proof.

\subsection{De-randomized algorithm}
We now give a de-randomized version of the approximation algorithm. The algorithm
will iteratively find partial mappings of $V_G$ to $V_H$  and remove vertices
for which the mapping is defined. We want the LP to be valid even after we removed
some vertices from $V_G$ and $V_H$. To this end, we slightly modify the LP.
We new LP is slightly weaker than the original LP.

\rule{0pt}{12pt}
\hrule height 0.8pt
\rule{0pt}{1pt}
\hrule height 0.4pt
\rule{0pt}{6pt}

\noindent \textbf{LP Relaxation}

$$\begin{array}{lll}
\max & \displaystyle\sum_{\substack{u,v \in V_G\\p,q\in V_H}}  w_G(u,v) w_H(p,q) y_{upvq}&\\
&\sum_{p\in V_H} x_{up} \leq 1,& \text{for all }  u\in V_G;\\
&\sum_{u\in V_G} x_{up} \leq 1,& \text{for all }   p\in V_H;\\
&\sum_{u \in V_G} y_{upvq} \leq x_{vq}, & \text{for all }  v\in V_G, \;p,q\in V_H;\\
& \sum_{p\in V_H} y_{upvq} \leq x_{vq}, &\text{for all }  u,v\in V_G, \;q\in V_H;\\
&y_{upvq} = y_{vqup}, &\text{for all } u,v\in V_G,\; p,q\in V_H;\\
&x_{up}\in [0,1], &\text{for all }  u\in V_G, \; p\in V_H;\\
&y_{upvq}\in[0,1], &\text{for all } u\in V_G, \; p\in V_H.\\
\end{array}
$$
\rule{0pt}{1pt}
\hrule height 0.4pt
\rule{0pt}{1pt}
\hrule height 0.8pt
\rule{0pt}{12pt}

This LP is obtained from the original LP by replacing equalities
``$=$'' with inequalities ``$\leq$''
in the first four constraints. The integrality gap of the new LP is the same as
of the original LP. In fact, given a feasible solution $x^*$, $y^*$ of the
new LP we can always increase the values of some variables to get a feasible solution
$x^{**}$, $y^{**}$ of the original LP (then $x^{**}\geq x^*$ and
$y^{**}\geq y^*$ component-wise).

\begin{theorem}
There exists a polynomial time (deterministic) algorithm that given  an instance $\Gamma$ of \MAXQAP consisting of two weighted graphs $G=(V_G, w_G)$, $H=(V_H,w_H)$ and a solution $(x^*,y^*)$ to the LP,  of cost $LP^*$, outputs a bijection $\varphi : V_G \to V_H$ such that
$$\VALUEQAP (\Gamma, \varphi) \geq \frac{LP^*}{O(\sqrt{n})}.$$
\end{theorem}
\begin{proof}
The existence of the map $\varphi$ follows from Theorem~\ref{thm:analys}.
We have already established that either $LP^*_I \geq LP^*_{II}$ (see Theorem~\ref{thm:analys}
for definitions) and then
$$
\sum_{u,v\in V_G}\sum_{p,q\in V_H} \frac{x^{*}_{up}}{n} w_G(u,v) w_H(p,q)\geq C_{r.alg}\frac{LP^*}{\sqrt{n}};$$
or $LP^*_{II} \geq LP^*_{I}$ and then there exists a map $\varphi_{r.alg}:V_G \to V_H$
(returned by the randomized algorithm) and
a disjoint set of stars $\calS = \{(u, \calR_u)\}$
(each with the center in the vertex $u\in V_G$ and leaves $\calR_u\subset V_G$)
such that
$$\sum_{(u,\calR_u) \in \calS}
\sum_{v\in \calR_u}   w_G(u,v) w_H(\varphi_{r.alg}(u),\varphi_{r.alg}(v)) \geq C_{r.alg}\frac{LP^*}{\sqrt{n}},$$
for some universal constant $C_{r.alg}$. We consider these cases separately.

\medskip

I. First, assume that $LP_I^*\geq LP^*_{II}$. Our approach is similar to
the approach we used in Theorem~\ref{thm:analys}. However, instead of
peaking random sets $L_G$, $L_H$ and random maps $\varphi$ and $\nu$
we pick them deterministically. We first find $\varphi$ and $\nu$
to maximize the fractional value:
$$\sum_{u\in V_G}\sum_{v\in V_G} w_G(u,v) w_H(\varphi(u),\nu(v)).$$
Then, we pick $L_G$ and $L_H$ greedily to maximize
$$\sum_{\substack{u\in L_G: \varphi(u)\in L_H\\v\in R_G: \nu(v)\in R_H}} w_G(u,v) w_H(\varphi(u),\nu(v)).$$
We map $L_G$ according to $\varphi$ and $R_G=V_G\setminus L_G$ according to $\nu$. The details are below.

Find a bijection  $\varphi:V_G\to V_H$ that maximizes
$$
\Exp_{\nu} \sum_{u\in V_G} \left[\sum_{\substack{v \in V_G}}
w_G(u,v) w_H(\varphi(u),\nu(v))\right]
=
\sum_{u\in V_G} \left[\frac{1}{n}\sum_{\substack{v \in V_G\\q\in V_H}}
w_G(u,v) w_H(\varphi(u),q)\right],
$$
here $\nu:V_G\to V_H$ is a random bijection chosen uniformly from the set of all bijections. We find the bijection by solving the maximum matching problem between
$V_G$ and $V_H$, where the cost of mapping $u\mapsto p$ equals
$$\frac{1}{n}\sum_{\substack{v \in V_G\\q\in V_H}}
w_G(u,v) w_H(p,q).$$
Then we find a bijection $\nu:V_G\to V_H$ that
maximizes
$$
\sum_{u,v\in V_G} w_G(u,v) w_H(\varphi(u),\nu(v)).
$$
Again, we do this by solving the maximum matching problem, where now the cost
of mapping  $v\mapsto q$ equals
$$
\sum_{u\in V_G} w_G(u,v) w_H(\varphi(u),q).
$$
Since for a random permutation $\nu_{I}$ the maximum is at least
$C_{r.alg}LP^*/\sqrt{n}$, we get
\begin{equation}\label{eq:phinu1}
\sum_{u\in V_G}\sum_{v\in V_G} w_G(u,v) w_H(\varphi(u),\nu(v))
\geq C_{r.alg} \frac{LP^*}{\sqrt{n}}.
\end{equation}
We now use the greedy deterministic MAX CUT approximation algorithm\footnote{
The greedy MAX CUT algorithm picks vertices from the set $V_G$ in an arbitrary
order and puts them in the sets $L_G$ or $R_G$. Thus, at every step $t$
all vertices are partitioned into three groups $L_G(t)$, $R_G(t)$ and a group
of not yet processed vertices $U_G(t)$. If the weight of edges going from $v$ to
$R_G(t)$ is greater than the weight of edges going from $v$ to $L_G(t)$, then
the algorithm adds $v$ to $L_G$, otherwise to $R_G$. The algorithm maintain
the following invariant: at every step the weight of cut edges is greater than
or equal to the weight of uncut edges. Thus, in the end, the weight of cut edges
is at least a half of the total weight of all edges.}
to partition
$V_G$ into two sets $L_G$ and $R_G$ so as to maximize
$$\sum_{u\in L_G}\sum_{v\in R_G} w_G(u,v) w_H(\varphi(u),\nu(v)).$$
The cost of cutting an edge $(u,v)$ is $w_G(u,v) w_H(\varphi(u),\nu(v))$.
The cost of the obtained solution is at least a half of (\ref{eq:phinu1}).
We now use the greedy deterministic MAX DICUT (directed maximum cut) approximation algorithm\footnote{
The greedy MAX DICUT algorithm first finds an undirected maximum cut
$(A_G, B_G)$ using the greedy MAX CUT algorithm. The cost of the undirected maximum cut
is at least a half of the total weight of all edges. Then, it
outputs the cut $(A_G, B_G)$, if more edges are directed
from $A_G$ to $B_G$ than from $B_G$ to $A_G$, it
outputs the cut $(B_G, A_G)$, otherwise.
The cost  of the directed cut is at least a quarter of the total weight of all directed edges.}
to partition $V_H$ into
sets $L_H$ and $R_H$ so as to maximize
$$\sum_{\substack{u\in L_G\\\varphi(u)\in L_H}}
\sum_{\substack{v\in R_G\\\varphi(v)\in R_H}} w_G(u,v) w_H(\varphi(u),\nu(v))
=
\sum_{\substack{p\in L_H \\ \varphi^{-1}(p)\in L_G}}
\sum_{\substack{q\in R_H\\\nu^{-1}(q)\in L_H}} w_G(\varphi^{-1}(p),\nu^{-1}(q)) w_H(p,q).$$
The cost of a directed edge $(p,q)$ is $w_G(\varphi^{-1}(p),\nu^{-1}(q)) w_H(p,q)$, if
$\varphi^{-1}(p)\in L_G$, $\nu^{-1}(q)\in R_G$; and $0$ otherwise.
The cost of the obtained solution is at least $1/8$ of
(\ref{eq:phinu1}). Thus
\begin{equation}\label{eq:derand1}
\sum_{\substack{u\in L_G: \varphi(u)\in L_H\\v\in R_G: \nu(v)\in R_H}} w_G(u,v) w_H(\varphi(u),\nu(v))
\geq \frac{C_{r.alg}}{8}\frac{LP^*}{\sqrt{n}}.
\end{equation}
Note that we do not require that $|L_G| = |L_H|$ or that $|R_G| = |R_H|$. We output the map that maps $u\in L_G$ to $\varphi(u)$ if $\varphi(u)\in L_H$; and $v\in R_G$ to $\nu(v)$ if $\nu(v)\in R_H$. It maps the remaining vertices
in an arbitrary way. The cost of the solution is at least~(\ref{eq:derand1}).

\medskip

II. We now assume that there exists a collection of disjoint stars
$\calS = \{(u, \calR_u)\}$
(each with the center in the vertex $u\in V_G$ and leaves $\calR_u\subset V_G$)
and a map $\varphi_{r.alg}:V_G \to V_H$
such that
\begin{equation}\label{eq:stars1}
\sum_{(u,\calR_u) \in \calS}
\sum_{v\in \calR_u}   w_G(u,v) w_H(\varphi_{r.alg}(u),\varphi_{r.alg}(v)) \geq C_{r.alg}\frac{LP^*}{\sqrt{n}}.
\end{equation}

Define the LP volume of sets $S\subset V_G$, $T\subset V_H$ as follows:
$$
\vol_{LP}(S,T) =
\sum_{\substack{u\in S\\v\in V_G}}
\sum_{p,q\in V_H} w_G(u, v)w_H(p,q) y^*_{upvq}
+
\sum_{u, v\in V_G}
\sum_{\substack{p\in T\\q\in V_H}} w_G(u, v)w_H(p,q) y^*_{upvq}.
$$
If $S_1,\dots, S_k$ is a partition of $V_G$ and
$T_1,\dots, T_k$ is a partition of $V_H$, then
$$\sum_{i=1}^k \vol_{LP}(S_k,T_k) = 2LP^*,$$
since on the left hand side every term of the LP is counted twice.
Particularly,
$$\sum_{(u,\calR_u) \in \calS}
\vol_{LP}(\{u\}\cup \calR_u,\varphi_{r.alg} (\{u\}\cup \calR_u)) = 2LP^*.$$
Plugging in (\ref{eq:stars1}), we get
$$
\sum_{(u,\calR_u) \in \calS}\left(
2 \sum_{v\in \calR_u}   w_G(u,v) w_H(\varphi_{r.alg}(u),\varphi_{r.alg}(v)) - \frac{C_{r.alg}}{\sqrt{n}} \vol_{LP}(\{u\}\cup \calR_u,\varphi_{r.alg} (\{u\}\cup \calR_u))\right) \geq 0.$$
This inequality implies that there exists
one star $(u^*,\calR_{u^*})$ such that
$$
2\sum_{v\in \calR_{u^*}}  w_G(u^*,v) w_H(\varphi_{r.alg}(u^*),\varphi_{r.alg}(v)) \geq \frac{C_{r.alg}}{\sqrt{n}}\;\vol_{LP}(\{u^*\}\cup \calR_{u^*},\varphi_{r.alg} (\{u^*\}\cup \calR_{u^*})).$$
We find a star $(u^*,\calR^*)$ and an injective map
$\varphi: \{u\}\cup \calR\to V_H$ satisfying this inequality.
We do this as follows: For every choice of $u$ and $\varphi(u)$,
we solve the maximum partial matching problem where the cost of assigning
$v\mapsto q$ equals
\begin{multline*}
2w_G(u,v) w_H(\varphi(u),q) - \\ - \frac{C_{r.alg}}{\sqrt{n}}\left[
\sum_{\substack{u'\in V_G}}
\sum_{p',q'\in V_H} w_G(u', v)w_H(p',q') y^*_{u'p'vq'}
+
\sum_{u', v'\in V_G}
\sum_{\substack{p'\in V_H}} w_G(u', v')w_H(p',q) y^*_{u'p'v'q}
\right].
\end{multline*}
The set of matched vertices $v$ is the set of leaves of the star; $u$ is the center.

We fix the solution to be $\varphi$ on $(u^*,\calR^*)$. We remove the
star $(u^*,\calR^*)$ from the graph $G$
and its image $(\varphi(u^*), \varphi (\calR^*))$ from the graph $H$.
We repeat the algorithm recursively for the remaining graphs (we do not resolve the LP, but we again consider two cases: $LP_{I}^*\geq LP_{II}^{*}$ and $LP_{I}^*\leq LP_{II}^{*}$).
To estimate the cost of the solution, observe that the value of the LP decreases by
\begin{align*}
\sum_{\substack{u,v \in V_G\\p,q\in V_H}} &  w_G(u,v) w_H(p,q) y^*_{upvq}-
\sum_{\substack{u,v \in V_G\setminus{(\{u^*\}\cup \calR^*)}\\p,q\in V_H
\setminus ( \{\varphi(u^*)\}\cup \varphi(\calR^*) )}}  w_G(u,v) w_H(p,q)
y^*_{upvq}
\\&
\leq \phantom{+}\sum_{\substack{u\in (\{u^*\}\cup \calR^*) ,v \in V_G\\p,q\in V_H}}
 w_G(u,v) w_H(p,q) y^*_{upvq} +
\sum_{\substack{u,v \in V_G\\ p\in ( \{\varphi(u^*)\}\cup \varphi(\calR^*) ) ,q\in V_H}}
w_G(u,v) w_H(p,q) y^*_{upvq}
\\& \phantom{\leq}
+ \sum_{\substack{u \in V_G,v \in (\{u^*\}\cup \calR^*)\\ p,q\in V_H}}
w_G(u,v) w_H(p,q) y^*_{upvq} +
\sum_{\substack{u,v \in V_G\\p\in V_H, q \in ( \{\varphi(u^*)\}\cup \varphi(\calR^*) )}}
 w_G(u,v) w_H(p,q) y^*_{upvq}\\
&= 2\vol_{LP}(\{u^*\}\cup \calR^*,\varphi(\{u^*\}\cup \calR^*)),
\end{align*}
while the profit we get from mapping $(u^*,\calR^*)\mapsto (\varphi(u^*),\varphi(\calR^*))$ is at least
$$\frac{C_{r.alg}}{2\sqrt{n}} \vol_{LP}(\{u^*\}\cup \calR^*,\varphi(\{u^*\}\cup \calR^*)).$$
Hence, the approximation ratio is at least $C_{r.alg}/(4\sqrt{n})$.
\end{proof}

\section{Acknowledgement}
We would like to thank anonymous referees for valuable comments and suggestions.

\pagebreak

\appendix
\section{Appendix}
\noindent\textbf{Lemma~\ref{fact1}}\emph{
Let $S$ be a random subset of a set $V$. Suppose that for $u\in V$,
all events $\{u'\in S\}$ where $u'\neq u$ are jointly independent of
the event $\{u\in S\}$. Let $s$ be an element of $S$ chosen uniformly at random (if $S=\varnothing$, then $s$ is not defined). Then $\Prob{u=s} \geq  \Prob{u\in S}/(\E{|S|} + 1)$.}
\begin{proof}
We have
$$\Prob{u=s} = \Prob{u\in S} \times \E{\frac{1}{|S|} \Given u\in S}.$$
By Jensen's inequality $\E{1/|S| \given u\in S} \geq 1/\E{|S| \given u\in S}$. Moreover,
$$\E{|S| \given u\in S} = \E{|S\setminus \{u\}| \given u\in S} + 1 =
\E{|S\setminus \{u\}|} + 1 \leq \E{|S|} + 1.$$
\end{proof}

\noindent\textbf{Lemma~\ref{fact2}}\emph{
Let $S$ be a random subset of a set $L$, and $T$ be a random subset of a set $R$.
Suppose that for $(l,r)\in L\times R$, all events $\{l'\in S\}$ where $l'\neq l$ and all events $\{r'\in T\}$ where $r'\neq r$ are jointly independent of the event $\{(l,r) \in S\times T\}$.
Let $s$ be an element of $S$ chosen uniformly at random, and let $t$ be an element of $T$ chosen uniformly at random.
Then, $$\Prob{(l,r) = (s,t)} \geq \frac{\Prob{(l,r)\in S\times T}}{(\E{|S|} + 1)\times (\E{|T|} + 1)}$$
(here $(s,t)$ is not defined if $S=\varnothing$ or $T=\varnothing$).
}
\begin{proof}
We have
$$\Prob{(l,r) = (s,t)} = \Prob{(l,r)\in S\times T} \times \E{\frac{1}{|S|\cdot|T|}
 \Given (l,r)\in S\times T}.$$
Note, that if $(l,r)\in S\times T$, then $S\neq \varnothing$ and $T\neq \varnothing$ and hence $1/(|S|\cdot|T|)$ is well defined. By Jensen's inequality (for the convex function $t\mapsto (1/t)^2$),
\begin{multline*}
\E{\frac{1}{|S|\cdot|T|} \Given (l,r)\in S\times T}
=\\
\E{\left(\frac{1}{\sqrt{|S|\cdot|T|}}\right)^2 \Given (l,r)\in S\times T}
\geq \left(\frac{1}{\E{\sqrt{|S|\cdot|T|}\Given (l,r)\in S\times T}}\right)^2.
\end{multline*}
Then,
\begin{eqnarray*}
\E{\sqrt{|S|\cdot|T|} \Given (l,r)\in S\times T} &=&
\E{\sqrt{(|S\setminus\{l\}| + 1) (|T\setminus\{r\}| +1)} \Given (l,r)\in S\times T}\\
&=&
\E{\sqrt{(|S\setminus\{l\}| + 1) (|T\setminus\{r\}| +1)}} \\
&\leq&
\E{\sqrt{(|S| + 1) (|T| +1)}}\\
&\le&  \sqrt{\E{|S| + 1}\E{|T| +1}},
\end{eqnarray*}
where the last inequality follows from the  Cauchy-Schwarz  inequality. This finishes the proof.
\end{proof}


\begin{thebibliography}{BCCM}


\bibitem{AJ} W. P. Adams and T. A. Johnson, Improved Linear Programming-based Lower Bounds for the Quadratic
Assignment Problem, DIMACS Series in Discrete Mathematics and
Theoretical Computer Science, 16, 1994, 43-77.

\bibitem{A} K. Anstreicher, Recent advances in the solution
of quadratic assignment problems, ISMP, 2003 (Copenhagen). Math.
Program. 97, 2003, no. 1-2, Ser. B, 27-42.


\bibitem{AHS} E. Arkin, R. Hassin and M. Sviridenko, Approximating the Maximum Quadratic Assignment Problem,
Information Processing Letters, 77, 2001, pp. 13-16.

\bibitem{AFK} S. Arora, A. Frieze and H. Kaplan, A new rounding procedure for the assignment problem with
applications to dense graph arrangement problems, Mathematical
Programming, 92(1), 2002, 1-36.

\bibitem{PCP0} S. Arora and C. Lund. Hardness of Approximations.
In Approximation Algorithms for NP-hard Problems, Dorit Hochbaum, Ed.
PWS Publishing , 1996.

\bibitem{PCP1} S. Arora, C. Lund, R. Motwani, M. Sudan, and M. Szegedy,
Proof verification and the hardness of approximation problems, Journal of the
ACM 45 (3).

\bibitem{PCP2} S. Arora, and S. Safra, Probabilistic checking of proofs: A new characterization of NP, Journal of the ACM 45 (1): 70–122.

\bibitem{AKKV} V. Arvind, J. Kobler, S. Kuhnert, and Y. Vasudev. 
Approximate graph isomorphism. ITCS 2012.

\bibitem{B} A. Barvinok, Estimating $L^\infty$ norms by $L^{2k}$ norms for functions on orbits. Found. Comput. Math. 2 (2002), no. 4, 393--412.

\bibitem{BCCFV} A. Bhaskara, M. Charikar, E. Chlamtac, U. Feige and A. Vijayaraghavan,
Detecting High Log-Densities -- an $O(n^{\nicefrac{1}{4}})$ Approximation for Densest $k$-Subgraph,
to appear in Proceedings of STOC 2010.


\bibitem{bcpp} R.E. Burkard, E. Cela, P. Pardalos and L.S. Pitsoulis, The quadratic
assignment problem, In Handbook of Combinatorial Optimization, D.Z.
Du, P.M. Pardalos (Eds.), Vol. 3, Kluwer Academic Publishers, 1998,
241-339.

\bibitem{BDM} R.E. Burkard, M. Dell'Amico, S. Martello, Assignment Problems, SIAM Philadelphia, 2009.

\bibitem{CHK} M. Charikar, M. Hajiaghayi, H. Karloff, Improved Approximation Algorithms for Label Cover Problems, In Proceedings of ESA 2009, pp. 23--34

\bibitem{c} Eranda Cela, The Quadratic Assignment Problem: Theory and Algorithms,
Springer, 1998.


\bibitem{DW} Y. Dong and H. Wolkowicz,A Low-Dimensional Semidefinite Relaxation for the Quadratic Assignment Problem, Mathematics of Operations Research 34 (2009), pp. 1008-1022.

\bibitem{DWWZ} R. O'Donnell, J. Wright, C. Wu, Y. Zhou.
Hardness of Robust Graph Isomorphism, Lasserre Gaps, and Asymmetry of Random Graphs. SODA 2014.



\bibitem{DH} J. Dickey and J. Hopkins, Campus building arrangement using
TOPAZ, Transportation Science, 6, 1972, pp.~59--68.


\bibitem{EL} H. Eiselt and G. Laporte, A combinatorial
optimization problem arising in dartboard design, Journal of
Operational Research Society, 42, 1991, pp.~113--118.


\bibitem{E} A. Elshafei, Hospital layout as a quadratic assignment
problem, Operations Research Quarterly, 28, 1977, pp.~167--179.

\bibitem{F02} U. Feige, Relations between average case complexity and approximation complexity,
in Proceedings of STOC 2002, pp.~534-543.

\bibitem{F} U. Feige, private communication, 2009.

\bibitem{FK} A. Frieze and R. Kannan,
Quick approximation to matrices and applications.
Combinatorica 19 (1999), no. 2, 175--220.

\bibitem{GG} A. Geoffrion and G. Graves,  Scheduling parallel
production lines with changeover costs: Practical applications of
a quadratic assignment/LP approach, Operations Research, 24
(1976),  596-610.

\bibitem{hls} R. Hassin, A. Levin and M. Sviridenko, Approximating the minimum
quadratic assignment problems, ACM Transactions on Algorithms 6(1), (2009).

\bibitem{Khot} S. Khot, Ruling out PTAS for graph min-bisection, densest subgraph and bipartite
clique, In Proceedings of FOCS 2004, pp.~136–145.

\bibitem{KB} T. C. Koopmans  and M. Beckman,    Assignment problems and the location of economic activities, Econometrica 25 (1957), pp. 53–76.

\bibitem{LM} G. Laporte and H. Mercure, Balancing hydraulic
turbine runners: A quadratic assignment problem, European Journal
of Operations Research, 35 (1988), 378-381.

\bibitem{labhq} E.M. Loilola, N.M.M. De Abreu, P.O. Boaventura-Netto, P.M. Hahn,
and T. Querido, A survey for the quadratic assignment problem,
Invited Review, European Journal of Operational Research, 176,
657-690, 2006.

\bibitem{MDiar} C. McDiarmid,  Concentration. Probabilistic methods for algorithmic discrete mathematics, 195--248, Algorithms Combin., 16, Springer, Berlin, 1998.

\bibitem{NS} V. Nagarajan and M. Sviridenko, On the maximum quadratic assignment problem,
 Mathematics of Operations Research 34(4), pp. 859-868 (2009), preliminary version appeared in Proceedings of SODA 2009, pp. 516-524.

\bibitem{pw} P. Pardalos and H. Wolkowitz, eds., Proceedings of the DIMACS Workshop on Quadratic Assignment
Problems, DIMACS Series in Discrete Mathematics and Theoretical
Computer Science, 16, 1994.




\bibitem{PY} C. H. Papadimitriou and M. Yannakakis, The traveling salesman problem with distances one and two,  Mathematics of Operations Research  18 (1993), pp. 1 - 11.

\bibitem{Raz} R. Raz, A Parallel Repetition Theorem,
SIAM Journal on Computing, 27, 1998, pp.~763--803.


\bibitem{Q}
M. Queyranne, Performance ratio of polynomial heuristics for
triangle inequality quadratic assignment problems,  Operations
Research Letters, 4, 1986, 231-234.

\bibitem{S} L. Steinberg, The backboard wiring problem: a placement
algorithm. SIAM Rev. 3 1961 37--50.



\bibitem{ZKRW} Qing Zhao, Stefan E. Karisch, Franz Rendl and Henry Wolkowicz, Semidefinite Programming Relaxations for the Quadratic Assignment Problem, Journal Journal of Combinatorial Optimization 2(1998), pp. 71-109.

\end{thebibliography}
\end{document}